\documentclass[12pt,a4paper]{amsart}

\usepackage{amscd,amssymb,amsthm}
\usepackage{graphicx}

\def\ds{\displaystyle}
\numberwithin{equation}{section}

\theoremstyle{plain}
\newtheorem{theorem}{Theorem}[section]
\newtheorem{corollary}[theorem]{Corollary}
\newtheorem{lemma}[theorem]{Lemma}
\newtheorem{proposition}[theorem]{Proposition}

\theoremstyle{definition}

\theoremstyle{remark}

\numberwithin{equation}{section}

\numberwithin{table}{section}

\numberwithin{figure}{section}

\def\ds{\displaystyle}

\newcommand{\R}{\mathbb{R}}

\newcommand{\Z}{\mathbb{Z}}

\newcommand{\bea}{\begin{eqnarray*}}
\newcommand{\eea}{\end{eqnarray*}}
\newcommand{\bean}{\begin{eqnarray}}
\newcommand{\eean}{\end{eqnarray}}

\begin{document}

\title[New addition formula for the little $q$-Bessel functions]{New addition formula for the little $q$-Bessel functions}
\author{Fethi BOUZEFFOUR}
\address{Department of mathematics, College of Sciences\\ King Saud University,
 P. O Box 2455 Riyadh 11451, Saudi Arabia.} \email{fbouzaffour@ksu.edu.sa}

\subjclass[2000]{33D15,33D45, 33D80.}%
\keywords{Basic hypergeometric functions, Basic orthogonal polynomials and functions, Addition formula, Product formula.}%

\begin{abstract}
Starting from the addition formula for little $q$-Jacobi polynomials, we
derive a new addition formula for the little $q$-Bessel functions.
The result is obtained by the use of a limit transition. We also
establish a product formula for little $q$-Bessel functions with a positive and
symmetric kernel.
\end{abstract}
\maketitle
\section{Introduction}
The Bessel functions are defined by \cite{Wa}
\begin{equation}
J_{\alpha}(x)= \sum_{k=0}^{\infty}\frac{(-1)^k(z/2
)^{\alpha+2k}}{k!\Gamma(\alpha+k+1)}\label{1}.
 \end{equation}
These functions satisfy several addition formulas \cite{Wa, Vi, Erdely}. For example the Graf's addition formula
\begin{equation}
(\frac{y-s^{-1}x}{y-sx})^{\frac{n}{2}}J_{n}(\sqrt{(y-s^{-1}x)(y-sx)})=\sum_{k=-\infty}^{\infty}s^k
J_k(x)J_{n+k}(y)
\end{equation}
where $|ys^{\pm 1}|<|x|,$ for general $n$ and due to Neumann (1867) for $n=0$ (see \cite{Wa}).\\
Another addition formula obtained by Gegenbauer, is more connected
with the theory of spherical waves. It is
\begin{align}
&(x^2+y^2-2xy \cos(\phi))^{-\frac{1}{2}\alpha}J_{\alpha}(\sqrt{x^2+y^2-2xy \cos(\phi)})
\\&=2^{\alpha}\Gamma(\alpha)\sum_{k=0}^{\infty}(\alpha+k)x^{-\alpha}\nonumber
J_{\alpha+k}(x)y^{-\alpha}J_k(y)C_k^{\alpha}(\cos(\phi)),\label{222}
\end{align}
where $C_k^{\alpha}(\cos(\phi))$ are the ultraspherical polynomials
(see \cite{Sheme}).\\Formula (1.3) implies a product formula for Bessel functions, important in harmonic analysis and applied mathematics. It is related to the
generalized translation associated with the Bessel operator
introduced by J. Delsarte \cite{Delsarte}. Also, the product formula together
with its positivity plays a central role in the development of the
classical harmonic analysis.\\
Many $q$-analogues of the Bessel
functions exist. The oldest ones were introduced by Jackson
in 1903--1905. They were rewritten in modern notation by Ismail \cite{Ism}.
Other $q$-analogues can be obtained as formal limit cases of the
three $q$-analogues of Jacobi polynomials, i.e., of little $q$-Jacobi
polynomials, big $q$-Jacobi polynomials and Askey-Wilson
polynomials. For this reason we propose to speak about little
$q$-Bessel functions, big $q$-Bessel functions and AW type $q$-Bessel
functions respectively for the corresponding limit cases (see \cite{Askey-koelink}). Also observe that the little $q$-Bessel functions were already given by Jackson (1904), and that they are often called Jackson's third $q$-Bessel functions (see \cite{Ismail-Masson}). \\In this paper we are
concerned with the little $q$-Bessel functions defined by \cite{KoS}
\begin{equation} J_{\alpha}(z;q)=z^{\alpha}\frac{(q^{\alpha+1};q)_{\infty}}{(q;q)_{\infty}} \ _{1}\phi_1\left(\left. \begin{matrix} 0\\
q^{\alpha+1}\end{matrix}\right\vert q,z\right) \label{11}.
\end{equation}
It will be convenient to use the slightly different functions
$j_{\alpha}(z;q^2),$ called normalized little $q$-Bessel functions,
which are defined by \cite{FHB}
\begin{equation}
j_{\alpha}(z;q)=\ _{1}\phi_1\left(\left. \begin{matrix} 0\\
q^{\alpha+1}\end{matrix}\right\vert q,z\right)\label{12}.
\end{equation} There are now several addition formulas available for little
$q$-Bessel functions. The first one is a $q$-analogue of Graf's
addition formula \cite{Sad}, established by R. F. Swarttouw by using
an analytic method. A special cases of this addition formula is a
$q$-analogue of Neumann's addition formula for little $q$-Bessel
functions $J_0$.\\
The second addition formula is also a $q$-extension of Graf's addition
formula. This formula has originally been discovered by Koelink
\cite{KQ} using the interpretation of the little $q$-Bessel
functions on the
quantum group of plane motions.\\
The third addition formula is a $q$-analogue of Gegenbauer's
addition formula derived by R. F. Swarttouw \cite{Sad}. The purpose
of this paper to point out that a new addition formula for the
little $q$-Bessel functions can be derived as a limiting case of an
addition formula for the little $q$-Jacobi polynomials established
by Floris and Koelink (see \cite{FlorisKo}, \cite{Floris}). Our
procedure is suggested by the observation that the normalized little
$q$-Bessel functions $j_{\alpha}(z;q)$ can be obtained as a
confluent limit of the little $q$-Jacobi polynomials.
\section{Notations and Preliminaries}
Let $ a , q \in \R $ such that $0<q<1$ and $n \in \Z_{+}. $ The
$q$-shifted factorial are defined by \cite{GR}
\begin{equation}(a ;
q)_{0}: =1 \ , \quad  (a ; q)_{n} := \ds\prod_{k=0}^{n-1}\left(1 -
aq^{k}\right).
\end{equation}
For negative subscripts the $q$-shifted factorial is defined by
\begin{equation}
(a;q)_{-n}:=\frac{1}{(aq^{-n};q)_n},\,\ n=0,1,2,...\ \label{13}.
\end{equation}
We also define
\begin{equation}
(a;q)_{\infty}:=\prod_{k=0}^{\infty}(1-aq^k).
\end{equation} Recall that, for $n,\ m\in \mathbb{Z}_+$ and $\alpha>0,$ we have \\
 \begin{equation}
 (q^{\alpha};q)_{\infty}\leq(q^{\alpha};q)_n\leq (-q^{\alpha};q)_{\infty},\label{inq1}
 \end{equation}
 \begin{equation}
 \mid(q^{-m};q)_n q^{nm}\mid \leq q^{\binom{n}{2}}\label{inq}.
 \end{equation}
The basic hypergeometric series are defined by
\begin{equation}
   \ _{r}\varphi_{s} \left(\left. \begin{matrix} \ a_{1}, a_{2}, ..., a_{r}\\
  \ b_1,b_{2},...,b_s\end{matrix}\right\vert q,z\right):=\ds\sum_{k=0}^{+\infty}\ds\frac{\left(a_{1}, a_{2},
. . . , a_{r} \ ; q \right)_{k}}{\left(b_{1}, b_{2}, . . . , b_{s} \
; q \right)_{k}\left(q;\ q \right)_{k}}
\left((-1)^{k}q^{\left(^{k}_{2}\right)}\right)^{1+s-r}z^{k}, \label{14}
\end{equation} where
\begin{equation*}\left(a_{1}, . . . , a_{r} \ ; q \right)_{n}:= \left(a_{1}; \
q\right)_{n} \ . \  . \ . \left(a_{r};\ q\right)_{n}.
 \end{equation*}In
particular, the $\ _{1}\varphi_{1}$ basic hypergeometric series
satisfies the relation
\begin{align}
(w;q)_{\infty} \ _{1}\phi_{\ 1}\left( \left. \begin{matrix} \ a\  \\
\ w\end{matrix}\right\vert q,z \right)=(z;q)_{\infty} \ _{1}\phi_{\ 1}\left( \left. \begin{matrix} \ az/w\  \\
\ z\end{matrix}\right\vert q,w\right). \label{symmetric}
\end{align}
\begin{proposition}
Let $n\in \mathbb{Z}$ and $a \in \mathbb{C}$. The basic
hypergeometric series $_{1}\phi_{1}$ satisfies the relation
\begin{equation} (q^{1-n};q)_{\infty} \
_{1}\phi_{1}\left(\begin{matrix} \ a \ \\ q^{1-n} \
\end{matrix};q,z\right)=
(-z)^nq^{\binom{n}{2}}(a;q)_n(q^{1+n};q)_{\infty} \
_{1}\phi_{1}\left(\begin{matrix} \ aq^n \ \\ q^{1+n} \
\end{matrix};q,q^nz \right) \label{15}.
\end{equation}
In addition,
\begin{equation*} |(q^{1-n};q)_{\infty} \
_{1}\phi_{1}\left(\begin{matrix} \ a \ \\ q^{1-n} \
\end{matrix};q,z\right)|\leq
|z|^nq^{\binom{n}{2}}(-q,-|a|,-|z|;q)_{\infty},\,\, n=0,1,2,...\, .
\end{equation*}
\end{proposition}
\begin{proof}
Let $n \in \mathbb{N}.$ From \eqref{14}, we can write
\begin{equation}
(q^{1-n};q)_{\infty}\,_{1}\phi_{1}\left(\begin{matrix} \ a \ \\ q^{1-n} \
\end{matrix};q,z\right)=\sum_{k=0}^{\infty}\frac{(a;q)_k(q^{1-n+k};q)_{\infty}}{(q;q)_k}(-1)^kz^k
q^{\binom{k}{2}} \label{16}.
\end{equation}
The first $n$ terms of the series vanish, so the summation on the right-hand side of \eqref{16} starts with $k=n$, we obtain after replacing
$k$ by $k+n$
\begin{equation*}
(q^{1-n};q)_{\infty}\,_{1}\phi_{1}\left(\begin{matrix} \ a \ \\ q^{1-n} \
\end{matrix};q,z\right)=(-1)^nq^{\frac{1}{2}n(n-1)}z^n(a;q)_n(q^{1+n};q)_{\infty}\sum_{k=0}^{\infty}\frac{(aq^n;q)_k(-1)^kz^kq^{nk}
q^{\binom{k}{2}}}{(q^{1+n};q)_{k}(q;q)_k}.
\end{equation*}
Thus, we have shown the relation \eqref{15} for $n\geq 0$. The case $n<0$
follows from the case $n>0$  by changing $z$ by $zq^{-n}$ and $a$ by
$aq^{-n}$ in \eqref{15} and using relation \eqref{13}.
\end{proof}
\section{limit transition with $q$-orthogonal polynomials}
Bessel functions \eqref{1} are limit cases of some orthogonal
polynomials. A well known limit transition is the one between the
normalized Jacobi polynomials and the Bessel functions
$J_{\alpha}(z)$, which we can rewritten as
\begin{equation}
\lim_{n\rightarrow\infty}\,
P_n^{(\alpha,\beta)}\left(1-\frac{x^2}{2n^2}\right)=2^{\alpha}\Gamma
(\alpha +1)x^{-\alpha}J_{\alpha}(x), \label{31}
\end{equation}
where
\begin{equation}
P_n^{(\alpha,\beta)}(x)=\ _{2}F_1\left(\begin{matrix} -n,n+\alpha+\beta+1 \\
\alpha+1\end{matrix}\,;\frac{1-x}{2}\right)\label{32}.
\end{equation}
The $q$-analogue of limit transition \eqref{31} starts with the
little $q$-Jacobi polynomials, which are defined by
\begin{equation}
p_n(x;q^{\alpha},q^{\beta};q)=\ _{2}\phi_1\left(\begin{matrix} \
q^{-n}, \
q^{\alpha+\beta+n+1} \\
q^{\alpha+1}\end{matrix};q,x\right) \label{33},
\end{equation}
and which satisfy the orthogonality relation
\begin{align*}
\frac{(q^{\alpha+1},q^{\beta+1};q)_{\infty}}{(q^{\alpha+\beta+2},q;q)_{\infty}}\sum_{k=0}^{\infty}
p_m(q^k;q^{\alpha},q^{\beta};q)&p_n(q^k;q^{\alpha},q^{\beta};q)q^{k(\alpha+1)}\frac{(q^{k+1};q)_{\infty}}{(q^{\beta+k+1})_{\infty}}\\&=
\frac{q^{n(\alpha+1)}(1-q^{\alpha+\beta+1})(q^{\beta+1},q;q)_n}{(1-q^{\alpha+\beta+2n+1})(q^{\alpha+1},q^{\alpha+\beta+1};q)_{n}}\delta_{m,n},
\end{align*}
where $\alpha >-1$ and $\beta>-1,$ see \cite{Sheme}. The
$q$-analogue of the limit \eqref{31} from the little $q$-Jacobi
polynomials $p_n(x;a,b;q)$ to the normalized little $q$-Bessel
functions is given in the following proposition.
\begin{proposition} For $\alpha>-1$ and $\beta>-1,$ we have
\begin{equation}
\lim_{n \rightarrow
\infty}p_n(q^nx;q^{\alpha},q^{\beta};q)=j_{\alpha}(x;q),
\end{equation}
uniformly on compact subsets of $\mathbb{C}.$
\end{proposition}
\begin{proof} From \eqref{33}, we can write
\begin{equation}
p_n(xq^n;q^{\alpha},q^{\beta};q)=\sum_{k=0}^{n}\frac{(q^{-n};q)_k(q^{\alpha+\beta
+n+1};q)_k}{(q^{\alpha+1},q;q)_k}x^kq^{kn}.
\end{equation}
Let $|x|\leq M$. By inequality \eqref{inq}, the summand of the
sum at the right hand side can be majorized by
\begin{equation*}
\frac{(-q^{\alpha+\beta+1};q)_{\infty}}{(q^{\alpha+1},q;q)_k} q^{\binom{k}{2}}M^k
\end{equation*}
and
\begin{equation}
\sum_{k=0}^{\infty}\frac{(-q^{\alpha+\beta+1};q)_{\infty}}{(q^{\alpha+1},q;q)_k} q^{\binom{k}{2}}M^k=(-q^{\alpha+\beta+1};q)_{\infty}
\ _{1}\phi_1\left( \left.\begin{matrix} 0\\
q^{\alpha+1}\end{matrix}\right\vert q,-M\right)<\infty \label{34}.
\end{equation}
On the other hand a simple calculation shows that
\begin{equation}
\lim_{n\rightarrow\infty}\frac{(q^{-n};q)_k(q^{\alpha+\beta
+n+1};q)_k}{(q^{\alpha+1},q;q)_k}x^kq^{kn}=\frac{q^{\binom{k}{2}}}{(q^{\alpha+1},q;q)_k}x^k \label{35}.
\end{equation}
Then the result follows from \eqref{34} and \eqref{35} by dominated
convergence.
\end{proof}
Another limit transition from orthogonal polynomials to the Bessel
functions, starts with the Krawtchouk polynomials \cite{KoS}. In
this work we consider the affine $q$-Krawtchouk polynomials. They are defined by \cite{Sheme}
\begin{equation}K_{z}\left(
q^{-x};t,N;q\right)= \ _{3}\varphi_{ 2}\left(  \left.
\begin{array}
[c]{ccc}%
q^{-z}, & q^{-x}, & 0\\
tq, & q^{-N} &
\end{array}
\right\vert q,q\right)
\end{equation}
where $0\leq z\leq N$ and $0<tq<1$. They satisfy the orthogonality
relation
\begin{align}
\sum_{x=0}^{N}\frac{(tq;q)_x(q;q)_N}{(q;q)_x(q;q)_{N-x}}(tq)^{-x}&K_{n}\left(
q^{-x};t,N;q\right)K_{m}\left(
q^{-x};t,N;q\right)\\&=(tq)^{n-N}\frac{(q;q)_n(q;q)_{N-n}}{(tq;q)_n(q;q)_{N}}\delta_{n,m}.
\end{align}
Next, we will establish a limit transition from affine
$q$-Krawtchouk polynomials to big $q$-Bessel functions, which are
defined by
\begin{align*}
\mathcal{J}_{\lambda}(x,a;q)=\ _{1}\phi_{1}\left( \left. \begin{matrix} \ x^{-1}\  \\
\ a\end{matrix}\right\vert q,-\lambda a x \right).
\end{align*}
\begin{proposition}Let $\ N, \ z,$ and $ l\in \mathbb{Z}_+,$ with $z+l\leq
N$ and $0<tq<1.$ Then
\begin{align*}
\lim_{m\rightarrow
\infty}(q^{-N-m};q^{-1})_{z+m}K_{z+m}(q^{x-N};t,N+m;q)=\mathcal{J}_{-q^{N-z+1}}(t^{-1}q^{x-N-1},tq;q).
\end{align*}
Furthermore
\begin{align*}
|K_{z}(q^{x-N};t,N;q)|\leq\frac{1}{(q;q)_{\infty}} \ _{1}\phi_{1}\left( \left. \begin{matrix} \ -tq^{N-\Re(x)+1}\  \\
\ tq\end{matrix}\right\vert q,-q^{\Re(x)-z+1} \right).
\end{align*}
\end{proposition}
\begin{proof}
From the following identity (see \cite{GR})
\[ \ _{3}\varphi_{ 2}\left(  \left.
\begin{array}
[c]{ccc}%
q^{-n}, & c/b, & 0\\
c, & cq/bz &
\end{array}
\right\vert q,q\right)
=\frac{(-1)^nq^{\binom{n}{2}}}{(cq/bz;q)_n}(cq/bz)^n \
_{2}\varphi_{1}\left( \left.
\begin{array}
[c]{ccc}%
q^{-n},& b\\
\,\,\,\, c&
\end{array}
\right\vert q,z\right),
\]
we can rewrite the affine Krawtchouk polynomials as
$$K_{z}\left(  q^{x-N};t,N;q\right)=\frac{1}{(q^{N};q^{-1})_z}\ _{2}\varphi_{1}\left( \left.
\begin{array}
[c]{ccc}%
q^{-z}, & tq^{N-x+1}\\
\, \, \ tq &
\end{array}
\right\vert q,q^{x+1}\right).$$ Thus
\begin{align*}
(q^{-N-m};q^{-1})_{z+m}K_{z+m}(q^{x-N};t,N+m;q)=
\sum_{n=0}^{z+m}\frac{(q^{-z-m},q^{N-x+1};q)_n}{(tq,q;q)_n}q^{n(x+m+1)}.
\end{align*}
Now, by \eqref{inq}, we see that
\begin{equation}
|\frac{(q^{-z-l},tq^{N-x+1};q)_n}{(tq,q;q)_n}q^{n(x+l+1)}|\leq
\frac{(-tq^{N-\Re(x)+1};q)_n}{(tq,q;q)_n}q^{\binom{n}{2}}q^{n(\Re(x)-z+1)}.
\end{equation}
On the other hand
\begin{equation}
\lim_{m\rightarrow\infty}\frac{(q^{-z-m},tq^{N-x+1};q)_n}{(tq,q;q)_n}q^{n(x+m+1)}=\frac{(tq^{N-x+1};q)_n}{(tq,q;q)_n}(-1)^nq^{\binom{n}{2}}q^{n(x-z+1)}.
\end{equation}
So, the proposition follows by dominated convergence.
\end{proof}
\section{addition formula for the little $q$-Bessel functions}
Let $x \in \mathbb{C},$ $t\in (0,q^{-1}),$ $\nu >0$ and $l,$ $m,$
$z,$ and $N \in \mathbb{Z}_+ $ with $z+l\leq N.$\\We consider
\begin{align}
r_{l,m}^{(\nu)}(x;q)=\left\{\begin{matrix}(xq;q)_{l-m}^{1/2} \, \,p_{m}(x,q^{\nu},q^{l-m};q)\, \, \, l\geq m\\ \, \, \\
\, \, (x;q^{-1})_{m-l}^{1/2}\, \,
p_{l}(xq^{l-m},q^{\nu},q^{m-l};q)\, \,l\leq m.\end{matrix}\right .
\end{align}
where $p_{m}$ are the little $q$-Jacobi polynomials. P. G. A. Floris
and H. T. Koelink obtained the following addition formula for the
polynomials $r_{l,m}^{(\nu)}(x;q)$
\begin{align}
(-1)^{l-m}\sqrt{(tq)^{m-l}\frac{(xq^{m-l+1};q)_{l-m}}{(q^{N+m-l+1};q)_{l-m}}}r_{l,m}^{(\nu)}(xq^{m-l};q)\widehat{K}_{z}(xq^{-N};t,N+m-l;q)
\label{product}
\end{align}
\begin{align*}&=\sum_{r=0}^l\sum_{s=0}^mc_{l,m;r,s}^{(\nu)}(q)
(-1)^{r-s}t^{1/2(r+s)}q^{1/2(r-s)}q^{z(r+s)}r_{l-r,m-s}^{(\nu+r+s)}(q^z;q)r_{l-r,m-s}^{(\nu+r+s)}(tq^z;q)\\&
\times
r_{r,s}^{(\nu-1)}(q^{N+m-l-z};q)\widehat{K}_{z+l-m+s-r}(xq^{-N};t,N;q),
\end{align*}
where
\begin{equation}
c_{l,l,r,s}^{(\nu)}(q)=\frac{1-q^{\nu+r+s+1}}{1-q^{\nu+1}}\frac{c_{l,l}^{(\nu)}(q)}{c_{l-r,l-s}^{(\nu+r+s)}(q)c_{r,s}^{(\nu-1)}(q)}
\label{cc1},
\end{equation}
\begin{equation}
c^{(\nu)}_{l,m}(q)=\frac{q^{m(\nu+1)}(1-q^{\nu+1})}
{1-q^{\nu+l+m+1}}\frac{(q;q)_l(q;q)_m}{(q^{\nu+1};q)_l(q^{\nu+1};q)_m}\label{cc2},
\end{equation}
and
\begin{equation}
\widehat{K}_{z}\left(  q^{x-N};t,N;q\right)  =\left(  -1\right)
^{z}\left( tq\right) ^{-z/2}\left(  \frac{\left( q^{N};q^{-1}\right)
_{z}\left( tq;q\right)  _{z}  }{\left(  q;q\right)
_{z}}\right)  ^{1/2}K_{z}\left(  q^{x-N};t,N;q\right).\end{equation}
Let $l,m, k \in\mathbb{Z},$ with $m\leq l$, then from the formula
\eqref{product} we obtain
\begin{align}
&p_{l}(x;q^{\nu},1;q)\widehat{K}_{z}(xq^{-N};t,N;q) \label{add}
\end{align}
\begin{align*}&=\sum_{r=0}^{l}\sum_{s=0}^{r}(1-\delta_{rs})c_{l,l,r,s}^{(\nu)}(q)
(-1)^{r-s}t^{1/2(r+s)}q^{1/2(r-s)}q^{z(r+s)}(q^z,tq^z;q^{-1})_{r-s}^{1/2}\\&
 \times
(q^{N-z+1};q)_{r-s}^{1/2}
p_{l-r}(q^{z+s-r},q^{\nu+r+s},q^{r-s};q)p_{l-r}(tq^{z+s-r},q^{\nu+r+s},q^{r-s};q)\\&
\times
p_{s}(q^{N-z},q^{\nu-1},q^{r-s};q)\widehat{K}_{z+s-r}(xq^{-N};t,N;q)\\&+\sum_{r=0}^{l}\sum_{s=0}^{r}c_{l,l,r,s}^{(\nu)}(q)
(-1)^{r-s}t^{1/2(r+s)}q^{1/2(r-s)}q^{r^2-s^2}q^{z(r+s)}(q^{z+1},tq^{z+1};q)_{r-s}^{1/2}\\&\times
(q^{N-z};q^{-1})_{r-s}^{1/2}
p_{l-r}(q^{z},q^{\nu+r+s},q^{r-s};q)p_{l-r}(tq^{z},q^{\nu+r+s},q^{r-s};q)\\&\times
p_{s}(q^{N-z+s-r},q^{\nu-1},q^{r-s};q)\widehat{K}_{z+r-s}(xq^{-N};t,N;q).
\end{align*}
The addition formula for the little $q$-Bessel functions can be
obtained formally as a limit case of the formula \eqref{add}.
\begin{theorem}
Let $x \in \mathbb{C},$ $t\in (0,q^{-1}),$ $\nu >0$ and $z,\,N \in
\mathbf{N}.$ Then the series $_{1}\phi_{1}$ satisfies the following
addition formula
\begin{equation}
\ _{1}\phi_1\left(\begin{matrix} \ 0\  \\
\ q^{\nu+1}\end{matrix};q,q^{x-l}\right)\ _{1}\phi_1\left(\begin{matrix} \ tq^{N-x+1}\  \\
\
tq\end{matrix};q,q^{x-z+1}\right) \label{41}
\end{equation}
\begin{align*}&=\sum_{r=1}^{\infty}\sum_{s=0}^{r}(1-\delta_{rs})q^{(r+s)(z+s-l+1)}t^r
\frac{(q^{N-z+1};q)_{r-s}^{1/2}(q^N;q^{-1})_z^{1/2}}{(q^N;q^{-1})_{z+s-r}^{1/2}}\frac{(q^{\nu};q)_r
(q^{\nu};q)_s}{(q;q)_r (q;q)_s}\\&
 \times
\ _{1}\phi_1\left(\begin{matrix} \ 0\  \\
\ q^{\nu+r+s+1}\end{matrix};q,q^{z+s-l}\right)\
_{1}\phi_1\left(\begin{matrix} \ 0 \   \\
q^{\nu+r+s+1}\end{matrix};tq^{z+s-l}\right)\\& \times
p_{s}(q^{N-z},q^{\nu-1},q^{r-s};q)\ _{1}\phi_1\left(\begin{matrix} \
tq^{N-x+1}\   \\
tq\end{matrix};q^{x-z-s+r+1}\right)\\&+\sum_{r=1}^{\infty}\sum_{s=0}^{r}q^{(s+r)(z+r-l)}(tq)^s
\frac{(q^{N-z};q^{-1})_{r-s}^{1/2}(q^N;q^{-1})_z^{1/2}}{(q^N;q^{-1})_{z+s-r}^{1/2}}\frac{(q^{\nu};q)_r
(q^{\nu};q)_s}{(q;q)_r (q;q)_s}
\\&\times
\ _{1}\phi_1\left(\begin{matrix} \ 0\  \\
\ q^{\nu+r+s+1}\end{matrix};q,q^{z+r-l}\right)\
_{1}\phi_1\left(\begin{matrix} \ 0 \   \\
q^{\nu+r+s+1}\end{matrix};tq^{z+r-l}\right)\\& \times
p_{s}(q^{N-z+s-r},q^{\nu-1},q^{r-s};q)\
_{1}\phi_1\left(\begin{matrix} \
tq^{N-x+1} \   \\
tq\end{matrix};q^{x-z-r+s+1}\right).
\end{align*}
\end{theorem}
In order to justify the limit transition from the little $q$-Jacobi
to the little $q$-Bessel functions we first prove the following
lemma.
\begin{lemma} For  $\nu >0$ and $N,\,z,\,r,\,s \in
\mathbf{N},$ we have
\begin{equation*}
|p_{s}(q^{N-z},q^{\nu-1},q^{r-s};q)|\leq
\frac{(-q,-q^{\nu},-q^{z-N};q)_{\infty}}{(q,q^{\nu},q^{\nu};q)_{\infty}}q^{s(N-z)}q^{-s(s-1)/2}.
\end{equation*}
\end{lemma}
\begin{proof} In the sum defining $p_{s}(q^{N-z},q^{\nu-1},q^{r-s};q)$
replace the summation index $k$ by $s-k$ and apply elementary
relations for the $q$-shifted factorial (see \cite{GR}) to obtain
\begin{align*}
p_{s}(q^{N-z},q^{\nu-1},q^{r-s};q)=(-q^{N-z})^sq^{-s(s-1)/2}\frac{(q;q)_s}{(q^{\nu+r-s};q)_s}
\sum_{k=0}^{s}\frac{(q^{\nu+r-s};q)_{2s-k}(-q^{N-z})^{-k}q^{k(k-1)/2}}{(q;q)_{s-k}(q^{\nu};q)_{s-k}(q;q)_k}.
\end{align*}
By \eqref{inq1} and $q$-binomial theorem (see \cite{GR}), we get
\begin{align*}
\mid\sum_{k=0}^{s}\frac{(q^{\nu+r-s};q)_{2s-k}(-q^{N-z})^{-k}q^{k(k-1)/2}}{(q;q)_{s-k}(q^{\nu};q)_{s-k}(q;q)_k}\mid
& \leq
\frac{(-q^{\nu};q)_{\infty}}{(q,q^{\nu};q)_{\infty}}\sum_{k=0}^{s}\frac{q^{-k(N-z)}q^{k(k-1)/2}}{(q;q)_k}
\\& \leq
\frac{(-q^{\nu},-q^{z-N};q)_{\infty}}{(q,q^{\nu};q)_{\infty}}.
\end{align*}
Hence
\begin{equation*}
|p_{s}(q^{N-z},q^{\nu-1},q^{r-s};q)|\leq
\frac{(-q,-q^{\nu},-q^{z-N};q)_{\infty}}{(q,q^{\nu},q^{\nu};q)_{\infty}}q^{s(N-z)}q^{-s(s-1)/2}.
\end{equation*}
\end{proof}
\begin{proof}($\mathrm{Theorem \,4.1}$)
In the addition formula for the little $q$-Jacobi polynomials $p_n$,
replace $l$ by $l+m,$ $z$ by $z+m,$ $x$ by $q^{x+m}$ and $N$ by
$N+m$, we get addition formula for $p_{l+m}(q^{x+m};q^{\nu},1;q)$
\begin{align*}
&p_{l+m}(q^{x+m};q^{\nu},1;q)\widehat{K}_{z+m}(q^{x-N};t,N+m;q)\\&=\sum_{r=0}^{l+m}\sum_{s=0}^{r}(1-\delta_{rs})c_{l+m,l+m,r,s}^{(\nu)}(q)
(-1)^{r-s}t^{1/2(r+s)}q^{1/2(r-s)}q^{z(r+s)}\\&
 \times
(q^{z+m},tq^{z+m};q^{-1})_{r-s}^{1/2}
(q^{N-z+1};q)_{r-s}^{1/2}
p_{l+m-r}(q^{z+m+s-r},q^{\nu+r+s},q^{r-s};q)\\&
 \times
p_{l+m-r}(tq^{z+s+m-r},q^{\nu+r+s},q^{r-s};q)
p_{s}(q^{N-z},q^{\nu-1},q^{r-s};q)\widehat{K}_{z+m+s-r}(q^{x-N};t,N+m;q)\\&+\sum_{r=0}^{l+m}\sum_{s=0}^{r}c_{l+m,l+m,r,s}^{(\nu)}(q)
(-1)^{r-s}t^{1/2(r+s)}q^{1/2(r-s)}q^{r^2-s^2}q^{(z+m)(r+s)}\\&
\times(q^{z+m+1},tq^{z+m+1};q)_{r-s}^{1/2}
(q^{N-z};q^{-1})_{r-s}^{1/2}
p_{l+m-r}(q^{z+m},q^{\nu+r+s},q^{r-s};q) \\&
\times p_{l+m-r}(tq^{z+m},q^{\nu+r+s},q^{r-s};q)
p_{s}(q^{N-z+s-r},q^{\nu-1},q^{r-s};q)\widehat{K}_{z+m+r-s}(q^{x-N};t,N+m;q).
\end{align*}
As $m\rightarrow \infty$, the right-hand side and left-hand sides of
these formula tends to the right-hand side and left-hand side of the
addition formula for $q$-Bessel \eqref{41}. We give a rigorous
proof of these limits by using the dominated
convergence theorem.\\
From \eqref{cc1} and \eqref{cc2} there exists a positive constant $K$
such that
\begin{equation}
c_{l+m,l+m,r,s}^{(\nu)}(q)\sim
Kq^{-(l+m)(r+s)+s(s+r+1)}\, \, \, as\,\,\,
m\rightarrow\infty.
\end{equation}
Then by the use of proposition 3.2 and lemma 4-2, the summand in the
first sum of the second hand side of next formula is bounded in
absolute value by $Cq^{-l(r+s)+1/2s^2}$, where $C$ is positive
constant and independent of $s,r$. Similarly, the summand in the
second sum at the second hand side of next formula is bounded in
absolute value by $Cq^{-l(r+s)+1/2r^2}.$
\end{proof}
In the following corollary we rewrite the addition formula in
theorem 4.1 in terms of the normalized little $q$-Bessel functions.
\begin{corollary} For $\nu >0,$ $t\in(0,q^{-1})$ and $N,\,x,\,z \in
\mathbb{N},$ The normalized little $q$-Bessel functions satisfy the
following addition formula
\begin{align}
j_{\nu}(q^x;q)\ _{1}\phi_1\left(\begin{matrix} \ tq^{N-x+1}\  \\
\
tq\end{matrix};q,q^{x-z+1}\right)\label{42}
\end{align}
\begin{align*}
&=\sum_{r=1}^{\infty}\sum_{s=0}^{r}(1-\delta_{rs})q^{(r+s)(z+s+1)}t^r
\frac{(q^{N-z+1};q)_{r-s}^{1/2}(q^N;q^{-1})_z^{1/2}}{(q^N;q^{-1})_{z+s-r}^{1/2}}\frac{(q^{\nu};q)_r
(q^{\nu};q)_s}{(q;q)_r (q;q)_s}\\&
 \times
j_{\nu+r+s}(q^{z+s};q)j_{\nu+r+s}(tq^{z+s};q)
p_{s}(q^{N-z},q^{\nu-1},q^{r-s};q)\ _{1}\phi_1\left(\begin{matrix} \
tq^{N-x+1}\   \\
tq\end{matrix};q^{x-z-s+r+1}\right)\\&+\sum_{r=1}^{\infty}\sum_{s=0}^{r}q^{(s+r)(z+r)}(tq)^s
\frac{(q^{N-z};q^{-1})_{r-s}^{1/2}(q^N;q^{-1})_z^{1/2}}{(q^N;q^{-1})_{z+s-r}^{1/2}}\frac{(q^{\nu};q)_r
(q^{\nu};q)_s}{(q;q)_r (q;q)_s}
\\&\times
j_{\nu+r+s}(q^{z+r};q)j_{\nu+r+s}(tq^{z+r};q)p_{s}(q^{N-z+s-r},q^{\nu-1},q^{r-s};q)\
_{1}\phi_1\left(\begin{matrix} \
tq^{N-x+1} \   \\
tq\end{matrix};q^{x-z-r+s+1}\right).
\end{align*}
\end{corollary}
Throughout the next formula is called addition formula for the
normalized little $q$-Bessel functions. Now if we let $N$ to
infinity in these formula we obtain for $\nu >0,$ $\mu >-1$ and
$N,\,x,\,z \in \mathbb{N}$
\begin{align}
j_{\nu}(q^x;q)j_{\mu}(q^{x-z+1};q)
\end{align}
\begin{align*}&
=\sum_{r=1}^{\infty}\sum_{s=0}^{r}(1-\delta_{rs})q^{(r+s)(z+s+1)}q^{\mu
r} \frac{(q^{\nu};q)_r(q^{\nu};q)_s}{(q;q)_r (q;q)_s}\\&
 \times
j_{\nu+r+s}(q^{z+s};q)j_{\nu+r+s}(q^{z+s+\mu};q)j_{\mu}(q^{x-z-s+r+1})
\\&+\sum_{r=1}^{\infty}\sum_{s=0}^{r}q^{(s+r)(z+r)}q^{s(1+\mu)}
\frac{(q^{\nu};q)_r
(q^{\nu};q)_s}{(q;q)_r (q;q)_s}
\\&\times
j_{\nu+r+s}(q^{z+r};q)j_{\nu+r+s}(q^{z+r+\mu};q)
j_{\mu}(q^{x-z-r-s+1};q).
\end{align*}
\section{Product formula}Taking $l=m+k,$ with $k \in \mathbb{N}$ in formula (6.2)
of \cite{FlorisKo} gives the following product formula for the
little $q$-Jacobi polynomials:
\begin{align}
&(q^z,tq^z;q)_{k}^{1/2}p_{m}(q^z;q^{\nu},q^{k};q)p_{m}(tq^z;q^{\nu},q^{k};q) \label{51}
\end{align}
\begin{align*}&=
(1-q^{\nu})\sum_{N=z+k}^{\infty}\sum_{x=0}^{N}(-1)^{k}\frac{(q^{x-k+1};q)_{k}}{(q^{N-k+1};q)_{k}^{1/2}}
\frac{(q^N;q^{-1})_{N-x}(tq;q)_{N-x}}{(q;q)_{N-x}} (tq)^{x-k/2}\\&
\times q^{\nu(N-z-k)}
 \widehat{K}_{z}\left(
q^{x-N};t,N-k;q\right)\widehat{K}_{z+k}\left(
q^{x-N};t,N;q\right)p_{m}(q^{x-k};q^{\nu},q^{k};q).
\end{align*}
Formally the product formula for the $_{1}\phi_{1}$ series can be
obtained as a limit case of the  product formula \eqref{51} for
little $q$-Jacobi polynomials. In \eqref{51} just replace $z$ by
$z+k$ and make the substitution $N\rightarrow N+m$ and $x\rightarrow
x+m$. and let $m$ to $\infty,$ we obtain the formula cited in the
following proposition.
\begin{proposition} For $\nu >0,$ and $x,\,y,\,z\,\in \mathbb{Z}$ with $x\leq y,$ we
have
\begin{align}
\ _{1}\phi_1\left(\begin{matrix} \ 0\  \\
\ q^{\nu+1}\end{matrix};q,q^{z-m}\right)\ _{1}\phi_1\left(\begin{matrix} \ 0\  \\
\ q^{\nu+1}\end{matrix};q,tq^{z-m}\right)\label{prod}
\end{align}
\begin{align*}&=(1-q^{\nu})
\frac{(tq;q)_{\infty}}{(q;q)_{\infty}}
\times \sum_{N=z+k}^{\infty}\sum_{x=-\infty
}^{N}\frac{(tq;q)_{N-x}
q^{\nu(N-z-k)}q^{(x-z)(1+y-x)}}{(q;q)_{N-x}(q^{N};q^{-1})_{z+k}^{1/2}(q^{N-k};q^{-1})_{z}^{1/2}(q^{N+1};q)_{\infty}^{1/2}(q^{N-k+1};q)_{\infty}^{1/2}}\\&\times
\,\ _{1}\phi_1\left(\begin{matrix} \ tq^{1+N-x-k}\  \\
\ tq\end{matrix};q,q^{1+x-z}\right)
\ _{1}\phi_1\left(\begin{matrix} \ tq^{1+N-x}\  \\
\ tq\end{matrix};q,q^{1+x-z-k}\right)\ _{1}\phi_1\left(\begin{matrix} \ 0\  \\
\ q^{\nu+1}\end{matrix};q,q^{x-k-m}\right)
\end{align*}
\end{proposition}
\begin{theorem} For $\nu >0,$ and $x,\,y,\,z\,\in \mathbb{Z},$ we
have the following product formula for the normalized little
$q$-Bessel functions

\begin{equation}
j_{\nu}(q^x;q)j_{\nu}(q^y;q)=\sum_{z=-\infty}^{\infty}\triangle_{\nu}(q^x,q^y,q^z;q)j_{\nu}(q^z;q),\label{besp}
\end{equation}
where
\begin{align}
&\triangle_{\nu}(q^x,q^y,q^z;q)= \label{delta}
\end{align}
\begin{align*}
\frac{(1-q^{\nu})}{(q;q)_{\infty}} \sum_{N}\frac{\,\left[(q^{1+y-x};q)_{\infty}\ _{1}\phi_1\left(\begin{matrix} \ q^{1+N+y-x-z}\  \\
\ q^{1+y-x}\end{matrix};q,q^{1+z-x}\right)
\right]^2
}{(q;q)_{N-z}(q^{N};q^{-1})_{x}(q^{N+1},q^{1+N+y-x-z};q)_{\infty}}
q^{\nu(N-x)+(1+y-x)(z-x)}.
\end{align*}
Furthermore, the kernel
\begin{equation}
q^{(\nu+1)z}\triangle_{\nu}(q^x,q^y,q^z;q),\label{kern}
\end{equation} is symmetric in $x,$
$y$ and $z.$
\end{theorem}
\begin{proof}The formula \eqref{besp} follows from \eqref{prod}
and \eqref{12}. To show the kernel \eqref{kern} is symmetric it
suffices to prove
\begin{equation}
q^{(\nu+1)x}\triangle(q^x,q^y,q^z;q)=q^{(\nu+1)z}\triangle(q^z,q^y,q^x;q),
\end{equation}
and
\begin{equation}
q^{(\nu+1)x}\triangle(q^x,q^y,q^z;q)=q^{(\nu+1)y}\triangle(q^y,q^x,q^z;q).\label{sym2}
\end{equation}
By application of \eqref{symmetric} and twice \eqref{15}, we can
write, for $x,\,y,\,z\, \in \mathbb{Z}$
\begin{align*}
&(q^{1+y-x};q)_{\infty}\ _{1}\phi_1\left(\begin{matrix} \ q^{1+N+y-x-z}\  \\
\ q^{1+y-x}\end{matrix};q,q^{1+z-x}\right)=(-1)^{x-z}q^{(x-z)(1+2y-x-z)/2}\\&(q^{1+N-x};q)_{x-z}(q^{1+y-z};q)_{\infty}
\ _{1}\phi_1\left(\begin{matrix} \ q^{1+N+y-x-z}\  \\
\ q^{1+y-z}\end{matrix};q,q^{1+x-z}\right),
\end{align*}
and by the following elementary relations
\begin{align*}
&(q;q)_{N-z}=(q;q)_{N-x}(q^{1+N-x};q)_{x-z},\\&
(q^{N};q^{-1})_{x}=(q^{N};q^{-1})_{z}(q^{1+N-x};q)_{x-z},
\end{align*}
we get
\begin{align*}
&q^{(\nu+1)x}\triangle(q^x,q^y,q^z;q)=\\&
\frac{(1-q^{\nu})}{(q;q)_{\infty}} \sum_{N}\frac{\,\left[(q^{1+y-z};q)_{\infty}\ _{1}\phi_1\left(\begin{matrix} \ q^{1+N+y-x-z}\  \\
\ q^{1+y-z}\end{matrix};q,q^{1+x-z}\right)
\right]^2
}{(q;q)_{N-x}(q^{N};q^{-1})_{z}(q^{N+1},q^{1+N+y-x-z};q)_{\infty}}
q^{\nu N+(y-z)(x-z)+x}
\end{align*}
Hence,
\begin{equation*}
q^{(\nu+1)x}\triangle(q^x,q^y,q^z;q)=q^{(\nu+1)z}\triangle(q^z,q^y,q^x;q)
\end{equation*}
For \eqref{sym2}, applying relation  \eqref{symmetric} and replacing
$N$ by $N+x-y$ in the second member of \eqref{delta}, the kernel
$\triangle(q^x,q^y,q^z;q)$ becomes
\begin{align*} &\triangle(q^x,q^y,q^z;q)=
\frac{1-q^{\nu}}{(q;q)_{\infty}} \sum_{N}\frac{(q^{1+N-z};q)^2_{x-y}}{(q;q)_{N+x-y-z}}\\&
\times \frac{\,\left[(q^{1+x-y};q)_{\infty}\ _{1}\phi_1\left(\begin{matrix} \ q^{1+N+x-y-z}\  \\
\ q^{1+y-x}\end{matrix};q,q^{1+z-x}\right) \right]^2
}{(q^{N+x-y};q^{-1})_{x}(q^{1+N+x-y},q^{1+N-z};q)_{\infty}}
q^{\nu(N-y)+(1+x-y)(z-y)}.
\end{align*}
The result follows from the elementary relations
\begin{align*}
&(q^{1+N-z};q)_{x-y}=\frac{1}{(q^{1+N+x-y-z};q)_{y-x}},\\&
(q;q)_{N+x-y-z}(q^{1+N+x-y-z};q)_{y-x}=(q;q)_{N-z},\\&
(q^{N+x-y};q^{-1})_{x}=\frac{(q^N;q^{-1})_y}{(q^{1+N+x-y};q)_{y-x}},\\&
(q^{1+N+x-y-z};q)_{y-x}(q^{1+N-z;q})_{\infty}=(q^{1+N+x-y-z};q)_{\infty}.
\end{align*}
\end{proof}


\begin{thebibliography}{99}
\bibitem{Askey}{R. Askey}, {Orthogonal Polynomials and Special Functions}, {CBMS-NSF Regional Conference Series Applied Math. $\mathbf{21}$, SIAM, Philadelphia PA, (1975).}
\bibitem{Jzero}{N. Bettaibi, F. Bouzeffour, H. B Elmonser, W. Binous}, {Elements of harmonic analysis related to the third basic zero order Bessel function}, {J. Math. Anal. Appl.
$\mathbf{342}$ (2008), 1203--1219.}
\bibitem{Delsarte} {J. Delsarte,} {Sur une extension de la formule de Taylor, J. Math. Pures Appl. (9) vol. 17 (1936) pp. 213--231}
\bibitem{FHB}{A. Fitouhi, M. Hamza and F. Bouzeffour}, {The $q$-$J_{\alpha}$ Bessel function},
{J. Approx. Theory, $\mathbf{115}$ (2002), 144--166.}
\bibitem{FV2}{R. Floreanini, and L. Vinet}, {Addition formulas for $q$-Bessel functions},
 {J. Math. Phys. $\mathbf{33}$ (1992), 2984--2988.}
\bibitem {Floris}{P. G. A. Floris}, {Addition formula for $q$-disk
polynomials,} {Compositio Mathematica (1994) $\mathbf{1}$,
123--149.}
\bibitem {FlorisKo}{P. G. A. Floris, and E. Koelink,} {A commuting $q$-analogue of the addition
formula for disk polynomials,} {Constructive Approximation (1995)
$\mathbf{13}$, 511--535.}
\bibitem{Erdely} {A. Erd\'elyi, W. Magnus, F. Oberhettinger and F. G. Tricomi,} { Higher Transcendental Functions, Vol. 2, McGraw-Hill, New York (1953), Chapter 10. MR 15, 419.}
\bibitem {GR}{G. Gasper and M. Rahman}, {Basic Hypergeometric Series,
2nd Edition (2004), Encyclopedia of Mathematics and Its
Applications}, { $\mathbf{96}$, Cambridge University Press,
Cambridge.}
\bibitem{FJ}{F. H. Jackson}, {On a $q$-Definite Integrals},{ Quarterly
Journal of Pure and Applied Mathematics $\mathbf{41}$ (1910),
193--203.}
\bibitem{Ism}{M. E. H. Ismail}, {The zeros of basic Bessel functions, the functions
$J_{\nu+ ax}(x)$, and associated orthogonal polynomials}, {J. Math.
Anal. Appl. $\mathbf{86}$ (1982), 1--19.}
\bibitem{Ismail-Masson}{M. E. H. Ismail, D. R. Masson and S. K. Suslov,} {The $q$-Bessel function on a $q$-quadratic grid,} {in: Algebric methods and $q$-speial functions, CRM Proc, Lecture Notes, 22, Amer. Math. Soc., (1999), pp. 183--200.}
\bibitem {kB}{H. T. Koelink}, {A basic analogue of Graf's addition formula and
related formulas}, {Integral Transforms and Special Functions
$\mathbf{1}$ (1993), 165--182.}
\bibitem {KQ}{H. T. Koelink,} {The quantum group of plane motions and the Hahn-Exton
$q$-Bessel function}, {Duke Math. J. $\mathbf{76}$ (1994),
483--508.}
\bibitem {KY}{H. T. Koelink,} {Yet another basic analogue of Graf's addition formula},
{J. Comp. Appl. Math  $\mathbf{68}$ (1986), 209--220}
\bibitem {KS}{H. T. Koelink, and R. F. Swarttouw}, {A $q$-analogue of Graf's addition formula for the Hahn-Exton
$q$-Bessel function}, {J. Approx. Theory $\mathbf{81}$ (1995),
260--273.}
\bibitem{Askey-koelink}{E. Koelink and J. V. Stokman}, {The Askey-Wilson function transform scheme, in: Special Functions 2000: Current Perspective and Future Directions, J. Bustoz et al. (ed.), NATO Science Series II, Vol. 30, Kluwer, Dordrecht, (2001), pp. 221--241; $\mathrm{arXiv:math/9912140v2}.$}
\bibitem{Sheme}{R. Koekoek, P. A. Lesky, R. F. Swarttouw},
{ Hypergeometric orthogonal polynomials and their $q$-analogues,
Springer Monographs in Mathematics}, {Springer-Verlag, Berlin,
(2010).}
\bibitem {KoS}{ T. H. Koornwinder and R. F.
Swarttouw,} { On $q$-analogues of the Fourier and Hankel
transforms,} {Trans. Amer. Math. Soc. $\mathbf{333}$, 1992,
445--461.}
\bibitem {Rad}{M. Rahman,} {An addition theorem and some product formulas for $q$-Bessel functions}, {Canad. J. Math.
$\mathbf{40}$ (1988), 1203--1221.}
\bibitem {Sad}{R. F. Swarttouw,} {An addition theorem and some product formulas for the Hahn-Exton $q$-Bessel functions,} {Canad. J. Math. $\mathbf{44}$ (1992), 867--879.}
\bibitem {Vi}{N. J. Vilenkin, and A. U. Klimyk,} {Representation of Lie Groups and Special Functions,} {3 volumes,
Kluwer, Dordrecht, (1991), (1993).}
\bibitem {Wa}{G. N. Watson}, {A Treatise on the Theory of Bessel Functions}, {2nd ed., Cambridge University Press,
Cambridge, (1944).}
\end{thebibliography}
\end{document}